\documentclass{SCIS2025}

\usepackage{graphicx}
\usepackage{cite}
\usepackage{picinpar}
\usepackage{amsmath}
\usepackage{url}
\usepackage{flushend}
\usepackage{colortbl}
\usepackage{soul}
\usepackage{multirow}
\usepackage{color}
\usepackage{hyperref}
\usepackage{enumerate}
\usepackage{siunitx}
\usepackage{breakurl}
\usepackage{epstopdf}
\usepackage{pbox}

\usepackage{amsmath,amssymb,amsthm}
\usepackage{hyperref}
\usepackage{booktabs}
\usepackage{color}
\usepackage{multicol}
\usepackage{makecell}
\usepackage{balance}
\usepackage{subcaption}
\usepackage{threeparttable}

\begin{document}

\ArticleType{RESEARCH PAPER}
\Year{2025}
\Month{January}
\Vol{68}
\No{1}
\DOI{}
\ArtNo{}
\ReceiveDate{}
\ReviseDate{}
\AcceptDate{}
\OnlineDate{}
\AuthorMark{}
\AuthorCitation{}

\title{Observer-Based Active Fault/Disturbance Compensation Control for Fully Actuated Systems}{Title for citation}

\author[1]{Weijie REN}{}
\author[1,2]{Guang-Ren DUAN}{{g.r.duan@hit.edu.cn, duangr@sustech.edu.cn}}
\author[1]{Ping LI}{}
\author[1]{He KONG}{}

\address[1]{Guangdong Provincial Key Laboratory of Fully Actuated System Control Theory and Technology, \\Southern University of Science and Technology, Shenzhen 518055, China}
\address[2]{Center for Control Theory and Guidance Technology, Harbin Institute of Technology, Harbin 150001, China}

\abstract{
This paper is concerned with fault/disturbance compensation control for fully actuated systems. In particular, we explore observer-based control, incorporating an active compensation mechanism. First, we propose a novel observer with enhanced design flexibility for the fully actuated system model, enabling simultaneous estimation of system states and exogenous unknown signals, such as faults or disturbances. Then, a nonlinear controller is developed with an active fault or disturbance compensation term, leveraging the fully actuated system approach. The asymptotic stability of both the state estimation error and the closed-loop control system is systematically established.
Finally, the feasibility and merits of the proposed method are validated through comparative simulations and experiments.
}


\maketitle

\section{Introduction}

In industrial processes, actuators and sensors are crucial components that ensure the precise operation of automated systems. However, faults or disturbances in these elements can severely disrupt their operation, leading to degraded performance, product quality issues, safety hazards, and significant economic losses \cite{Isermann1984Auto_Process}. 
To mitigate these risks, fault detection and fault-compensation control strategies are vital. Detection and compensation algorithms can dynamically adjust the control actions, ensuring continuity and resilience in the face of faults. This not only safeguards the process but also enhances the overall efficiency and reliability of industrial operations \cite{Mono-Isermann2006_fault}. In previous studies, numerous strategies have been proposed to manage faults or disturbances, including passive or active methods \cite{Mono-Ding_Model-based_FD,Xu2022TAC_Minimal}, observer-based or filter-based schemes \cite{Ren2023TFS_Guaranteed,WangLianPuigShen2025TCST}, etc.
For instance, Shen et al. \cite{Shen2023TSMC,Shen2023TAC} proposed fault detection methods based on iterative interval estimation for a specific class of discrete-time nonlinear and linear systems. Gu et al. \cite{Gu2024TFS} designed an effective dynamic guaranteed cost control method that is event-triggered and anti-disturbance, specifically targeting fuzzy systems, and demonstrated its application to wind turbine models.

Since the introduction of the state-space method \cite{Kalman1960_State-space,Kalman1960_Filter}, most linear and nonlinear control problems have followed a `standard' processing flow: first, transforming the system dynamic equations into a state-space model, which consists of several first-order differential equations, and then proceeding to controller design \cite{Mono-Freeman_2008_robust}. The state-space method has been popular and widely used for over half a century, with the exception of a few scholars in the field of robotics who directly employ high-order dynamic models. 
However, the state-space framework has inherent limitations, particularly due to its focus on system states. While it simplifies the study of state responses, estimation, and prediction, addressing control issues, especially for nonlinear systems, remains challenging. In many cases, designing control laws that achieve global stabilization is not possible, let alone assigning the expected closed-loop eigenstructure.
As highlighted in \cite{Liu2022JAS}, control methods based on first-order state-space models for networked nonlinear multi-agent systems are inherently restrictive. These methods often find it hard to realize global stabilization and consensus among agents, even in the absence of communication constraints. 

In contrast to the state-space method, Duan proposed the fully actuated system (FAS) approach \cite{Duan_IJSS_I_Models}, a high-order method for addressing control problems. The FAS approach fundamentally differs by transforming the dynamical model into a higher-order one through variable elimination. The resulting FAS model is centered on the control input rather than the state, which simplifies the design of control laws and a linear constant closed-loop system can be ultimately attained. The advantages of the FAS approach significantly simplify controller design for nonlinear systems \cite{Cui2023IJFS,Zhang2024RNC,Chen2024JFI,Yao2024TMech,ZhaoDuan2022SCIS}. As a result, the FAS approach has garnered much attention, leading to the investigation of various problems using this method, e.g., robust adaptive control \cite{Duan_IJSS_V_RobustAdaptive}, generalized PID \cite{Duan_IJSS_IX_PID}, optimal control \cite{Duan_IJSS_VIII_Optimal}, time-delay \cite{Wang2024TCSI}, event-triggered \cite{Wang2024TCSII}, and fault-tolerant control \cite{JiangDuanHou2024TCSI,RenDuanLiKong2025TMech}.

Recently, researchers have begun developing control strategies to address faults and disturbances using the FAS approach \cite{LiuYaoSunHanDuanKong2025TAES,Duan_IJSS_IV_Adaptive,Liu2022NeuroC,Cai2023Auto,Dong2023JFI,LiDuanZhangWangWang2025TIE}. In \cite{Duan_IJSS_IV_Adaptive}, adaptive stabilizing and tracking controllers are proposed, capable of estimating model uncertainties within the system and compensating for them through the designed controller. The robust version of this work is presented in \cite{Duan_IJSS_V_RobustAdaptive}, which extends the capabilities of \cite{Duan_IJSS_IV_Adaptive} to include disturbance attenuation. \cite{Liu2022NeuroC} introduces an intermediate observer-based adaptive fault-tolerant control framework for high-order FAS (HOFAS), effectively addressing the issue of unmet observation matching conditions. In \cite{Cai2023Auto}, a novel active fault tolerance framework for uncertain HOFASs is proposed. This work innovatively employs Lie derivatives to highlight the advantages of the uncertain faulty HOFAS model with nonlinear measurements, while also incorporating a dynamic data model into the active fault-tolerant control scheme. The integration of an adaptive observer and controller within the HOFAS framework ensures uniformly bounded stability, both theoretically and experimentally. Furthermore, \cite{Dong2023JFI} addresses the tracking control problem under full-state constraints by handling actuator faults using a nonlinear transformation function. Additionally, \cite{JiangDuanHou2024TCSI} leverages a proportional-integral (PI) observer to achieve control objectives for general mixed-order FAS, showing improved estimation and control performance.

Almost all related literature, including the aforementioned works, primarily addresses faults or disturbances in actuators, assuming sensor data to be consistently healthy and accurate \cite{Duan_IJSS_V_RobustAdaptive,JiangDuanHou2024TCSI,Duan_IJSS_IV_Adaptive,Liu2022NeuroC,Cai2023Auto,Dong2023JFI}. Such an assumption limits the practicality of fault-tolerance frameworks for FASs. 
To address this gap, this paper proposes an observer-based active fault/disturbance compensation control framework for continuous-time nonlinear systems, which are described by general mixed-order FASs. The proposed framework not only handles unexpected signals in actuators but also compensates for invasive or faulty signals in sensors. The contributions of this paper are summarized as follows:
\begin{enumerate}
    \item Different from existing observers \cite{JiangDuanHou2024TCSI,Liu2022NeuroC,Cai2023Auto} in the FAS framework, a novel observer with additional design parameters is proposed for mixed-order FAS models, offering greater design flexibility. The proposed observer can effectively estimate both system states and unexpected signals, such as faults and disturbances, in actuators and sensors.
    
    \item An active observer-based compensation control framework is developed, ensuring strict guarantees for the exponential convergence of the observer error and the closed-loop control system, representing an advancement beyond the work in \cite{JiangDuanHou2024TCSI}.
\end{enumerate}

The remainder of this paper is organized as follows. Section \ref{sec:Preliminaries} introduces the essential notations, presents the mixed-order FAS model, and formulates the design problems. Section \ref{sec:Active fault/disturbance compensated control} focuses on the design of the observer and observer-based compensation controller utilizing the FAS approach. 
Section \ref{sec:Simulation and experiment} details the validation through comparative simulation and experimental results. Section \ref{sec:Conclusion} summarizes the findings and suggests directions for future research.

\section{Preliminaries}		\label{sec:Preliminaries}

\subsection{Necessary notations}

In the paper, $\mathbb{R}^n$, $\mathbb{C}^{n}$, and $\mathbb{N}^{n}$ represent the $n$-dimensional Euclidean, complex, and natural number spaces, respectively. For a vector or matrix $A$, $A^{\mathrm{T}}$ denotes its transpose. Define $\mathbf{sym}(A) = A + A^{\mathrm{T}}$. The symbols $0$ and $I$ represent the all-zero and identity matrices, respectively, of appropriate dimensions.
For a complex scalar $x$, $\operatorname{Re}(x)$ provides its real part. For a matrix $A \in \mathbb{R}^{n\times n}$, $\lambda_i(A)$ represents its $i$-th eigenvalue, $i=1,2,\ldots,n$, and $\lambda_{\min}(A)$ denotes the minimum eigenvalue of $A$. The operator $\exp(t)$ signifies Euler's number raised to the power of the scalar $t$. The elements in $\ast$ of a matrix inequality can be inferred from the property of symmetry.
For $B_{i}\in \mathbb{R}, i = 1,2,\ldots,n$, we denote a new math operator
\[
\operatorname{blockdiag}(B_{i}) = 
\begin{bmatrix}
    B_1 & 0 & \cdots & 0 \\
    0 & B_2 & \cdots & 0 \\
    \vdots & \vdots & \ddots & \vdots \\
    0 & 0 & \cdots & B_n
\end{bmatrix}.
\]
For $A_{i}\in \mathbb{R}^{m\times m}, i = 1,2,\ldots,n$, $n_1\le n_2$, $n_1, n_2\in \mathbb{N}$, and $j\le k$, some frequently employed notations in the FAS approach are introduced as follows:
$
A_{0\sim n}=\left[ \begin{matrix}
    A_0&		A_1&		\dots&		A_n\\
\end{matrix} \right],
$
\[
x_{n_1\sim n_2}=\left[ \begin{array}{c}
    x_{n_1}\\
    x_{n_1+1}\\
    \vdots\\
    x_{n_2}\\
\end{array} \right],~
x^{\left( n_1\sim n_2 \right)}=\left[ \begin{array}{c}
    x^{\left( n_1 \right)}\\
    x^{\left( n_1+1 \right)}\\
    \vdots\\
    x^{\left( n_2 \right)}\\
\end{array} \right],~
x_{i}^{\left( n_i \right)}|_{i=j\sim k}=\left[ \begin{array}{c}
    x_{j}^{\left( n_j \right)}\\
    x_{j+1}^{\left( n_{j+1} \right)}\\
    \vdots\\
    x_{k}^{\left( n_k \right)}\\
\end{array} \right],~
x_{i}^{\left( n_0\sim n_i \right)}|_{i=j\sim k}=\left[ \begin{array}{c}
    x_{j}^{\left( n_0\sim n_j \right)}\\
    x_{j+1}^{\left( n_0\sim n_{j+1} \right)}\\
    \vdots\\
    x_{k}^{\left( n_0\sim n_k \right)}\\
\end{array} \right].
\]

\subsection{System description and problem formulation}

Consider the following general mixed-order FAS model:
\begin{align}   \label{eq:big FAS model}
    \begin{cases}
        \left[ \begin{array}{c}
            x_{1}^{\left( m_1 \right)}\\
            x_{2}^{\left( m_2 \right)}\\
            \vdots\\
            x_{\xi}^{\left( m_{\xi} \right)}\\
        \end{array} \right] =\left[ \begin{array}{c}
            f_1\left( x_{i}^{\left( 0\sim m_i-1 \right)}|_{i=1\sim \xi},\zeta ,t \right)\\
            f_2\left( x_{i}^{\left( 0\sim m_i-1 \right)}|_{i=1\sim \xi},\zeta ,t \right)\\
            \vdots\\
            f_{\xi}\left( x_{i}^{\left( 0\sim m_i-1 \right)}|_{i=1\sim \xi},\zeta ,t \right)\\
        \end{array} \right] +B\left( y,\zeta ,t \right) u+D_1(y)d\\
        y=Cx_{i}^{\left( 0\sim m_i-1 \right)}|_{i=1\sim \xi}+D_2d,\\
    \end{cases}
\end{align}
where $x_i \in \mathbb{R}^{r_i}$ for $i = 1, 2, \dots, \xi$ represents the system state, and $x_{i}^{(0 \sim m_i-1)} \in \mathbb{R}^{m_ir_i}$ and $x_{i}^{(m_i)} \in \mathbb{R}^{r_i}$ for $i = 1, 2, \dots, \xi$ correspond to the derivatives of orders $0$ through $m_i-1$ and $m_i$, respectively. Since \eqref{eq:big FAS model} contains mixed-order ordinary differential equations, we introduce two new symbols as follows:
\[
s=\sum_{i=1}^{\xi}{m_ir_i},~~r=\sum_{i=1}^{\xi}{r_i}.
\]
In the system, $u \in \mathbb{R}^r$ represents the control input, $y \in \mathbb{R}^p$ is the measured output, and $d \in \mathbb{R}^q$ appears in both the state and output equations, representing unknown time-varying faults or disturbances occurring in actuator and sensor components, respectively. The variable $\zeta\in \mathbb{R}^{n_{\zeta}}$ denotes external factors, and $t$ is the continuous time variable. The function $f_i\left( x_i^{(0 \sim m-1)}, \zeta, t \right) \in \mathbb{R}^{r_i}$ is a sufficiently differentiable nonlinear vector function, $B\left( y, \zeta, t \right) \in \mathbb{R}^{r \times r}$ is a matrix function, $C \in \mathbb{R}^{p \times s}$ is the output matrix, and $D_1(y) \in \mathbb{R}^{r \times q}$ and $D_2 \in \mathbb{R}^{p \times q}$ are full-column rank coefficient matrices.
The mixed-order FAS \eqref{eq:big FAS model} can be rewritten in a compact form:
\begin{align}   \label{eq:compact FAS model}
    \left\{ \begin{array}{l}
        x_{i}^{\left( m_i \right)}|_{i=1\sim \xi}=
        f\left( x_{i}^{\left( 0\sim m_i-1 \right)}|_{i=1\sim \xi},\zeta ,t \right) 
        +B\left( y,\zeta ,t \right)  u+D_1(y)d\\
        y=Cx_{i}^{\left( 0\sim m_i-1 \right)}|_{i=1\sim \xi}+D_2d,
    \end{array} \right. 
\end{align}
where $x_{i}^{\left( 0\sim m_i-1 \right)}|_{i=1\sim \xi}\in \mathbb{R}^s$, and
\begin{align*}
    f\left( x_{i}^{\left( 0\sim m_i-1 \right)}|_{i=1\sim \xi},\zeta ,t \right) 
    = \left[ \begin{array}{c}
        f_1\left( x_{i}^{\left( 0\sim m_i-1 \right)}|_{i=1\sim \xi},\zeta ,t \right)\\
        f_2\left( x_{i}^{\left( 0\sim m_i-1 \right)}|_{i=1\sim \xi},\zeta ,t \right)\\
        \vdots\\
        f_{\xi}\left( x_{i}^{\left( 0\sim m_i-1 \right)}|_{i=1\sim \xi},\zeta ,t \right)\\
    \end{array} \right].
\end{align*}

\begin{remark}	\label{rmk: B(y) explanation}
    It should be noted that in \eqref{eq:compact FAS model}, the input matrix function $B$ is assumed to depend on the system output $y$, rather than directly on the state vector $x_{i}^{\left( m_i \right)}|_{i=1\sim \xi}$. Addressing the case where $B(x_{i}^{\left( 0\sim m_i-1 \right)}|_{i=1\sim \xi},\zeta ,t)$ presents more complex procedures in observer design and stability analysis, which will be discussed elsewhere.
\end{remark}

For the FAS model, the following full-actuation assumption naturally holds.
\begin{assumption}	\label{asp:full-actuation condition}
    The system \eqref{eq:big FAS model} and \eqref{eq:compact FAS model} satisfies the full-actuation condition, i.e., for all $ y\in \mathbb{R}^{p}, \zeta\in \mathbb{R}^{n_{\zeta}}$, and $t\ge 0$, $\det\left( B\left( y,\zeta,t\right) \right) \ne 0 \text{ or } \infty $.
\end{assumption}

The state-space representation of FAS \eqref{eq:compact FAS model} is given by
\begin{align*}
    \dot{x}_{i}^{\left( 0\sim m_i-1 \right)}
    &=\Phi _i\left( 0_{0\sim m_i-1} \right) x_{i}^{\left( 0\sim m_i-1 \right)}
    +M_{ir}f_i\left( x_{i}^{\left( 0\sim m_i-1 \right)},\zeta,t \right)\\
    &+M_{ir}B\left( y,\zeta ,t \right) u+M_{ir}D_1(y)d, ~i=1,2,...,\xi,
\end{align*}
where
\[
\Phi _i\left( 0_{0\sim m_i-1} \right) =\left[ \begin{matrix}
    0&		I_{r_i}&		\cdots&		0\\
    0&		0&		\ddots&		\vdots\\
    \vdots&		\vdots&		\cdots&		I_{r_i}\\
    0&		0&		\cdots&		0\\
\end{matrix} \right] \in \mathbb{R}^{m_ir_i\times m_ir_i},~
M_{ir}=\left[ \begin{array}{c}
    0_{\left( m_i-1 \right) r_i\times r_i}\\
    I_{r_i}\\
\end{array} \right] \in \mathbb{R}^{m_ir_i\times r_i}.
\]
We can also express this in a compact form:
\begin{align}	\label{eq:state-space form of original FAS}
    \left\{
    \begin{array}{l}
        \begin{array}{ll}
            \dot{x}_{i}^{\left( 0\sim m_i-1 \right)}|_{i=1\sim \xi}
            =&\Phi _E(0)x_{i}^{\left( 0\sim m_i-1 \right)}|_{i=1\sim \xi}
            +M_Ef\left( x_{i}^{\left( 0\sim m_i-1 \right)}|_{i=1\sim \xi},\zeta ,t \right)\\
            &+M_EB\left( y,\zeta ,t \right)  u+M_ED_1(y)d
        \end{array}\\
        y=Cx_{i}^{\left( 0\sim m_i-1 \right)}|_{i=1\sim \xi}+D_2d,
    \end{array}
    \right.
\end{align}
where $\Phi _E\left( 0 \right) =\text{blockdiag}\left( \Phi _i\left( 0_{0\sim m_i-1} \right) \right) \in \mathbb{R}^{s \times s}$, $M_E=\text{blockdiag}\left( M_{ir} \right) \in \mathbb{R}^{s \times r}$, $i=1,2,...,\xi$.

To facilitate the development of an observer-based control framework, we make the following assumptions.

\begin{assumption}	\label{asp:observability and d's dimension constraint}
    The matrix pair $(\Phi _E(0),C)$ is observable, or at least, detectable. Assume that the dimension of the fault/disturbance vector $d$ is less than or equal to the dimension of the system output vector $y$. That is, $q \leq p$.
\end{assumption}

\begin{assumption}  \label{asp:Lipshitz condition}
    The nonlinear function in system \eqref{eq:compact FAS model} is sufficiently differentiable and the condition
    \begin{align*}
        &\lVert f\left( x_{i}^{\left( 0\sim m_i-1 \right)}|_{i=1\sim \xi},\zeta ,t \right) -f\left( \hat{x}_{i}^{\left( 0\sim m_i-1 \right)}|_{i=1\sim \xi},\zeta ,t \right) \rVert \\
        &\le \gamma _f\lVert x_{i}^{\left( 0\sim m_i-1 \right)}|_{i=1\sim \xi}-\hat{x}_{i}^{\left( 0\sim m_i-1 \right)}|_{i=1\sim \xi} \rVert 
    \end{align*}
    holds, where $\gamma_{f} > 0$ is the Lipschitz constant. Specifically, each nonlinear sub-function is Lipschitz continuous with a constant coefficient $\gamma_{f_i}>0$,
    \begin{align*}
        &\lVert f_i\left( x_{i}^{\left( 0\sim m_i-1 \right)}|_{i=1\sim \xi},\zeta ,t \right) -f_i\left( \hat{x}_{i}^{\left( 0\sim m_i-1 \right)}|_{i=1\sim \xi},\zeta ,t \right) \rVert \\
        &\le \gamma_{f_i}\lVert x_{i}^{\left( 0\sim m_i-1 \right)}|_{i=1\sim \xi}-\hat{x}_{i}^{\left( 0\sim m_i-1 \right)}|_{i=1\sim \xi} \rVert, ~i=1,2,...,\xi.
    \end{align*}
\end{assumption}

\begin{remark}
    Assumption \ref{asp:full-actuation condition} may seem restrictive, but it is actually rational because many systems have been rigorously proven equivalent to FASs based on the FAS approach. More details can be found in \cite{Duan_IJSS_I_Models,Duan_IJSS_II_strict-feedback} and other sources. Assumptions \ref{asp:observability and d's dimension constraint} and \ref{asp:Lipshitz condition} are also reasonable and are often necessary in observer-based fault estimation and the analysis of nonlinear systems.
\end{remark}

In order to handle the fault/disturbance signal $d$, the system \eqref{eq:state-space form of original FAS} can be equivalently transformed into a descriptor system \cite{Ren2023TFS_Guaranteed,WangLianPuigShen2025TCST}, as
\begin{align}   \label{eq:descriptor FAS}
    \left\lbrace 
    \begin{array}{l}
        E\dot{\tilde{x}}_{i}^{\left( 0\sim m_i-1 \right)}|_{i=1\sim \xi}
        =\tilde{P}\tilde{x}_{i}^{\left( 0\sim m_i-1 \right)}|_{i=1\sim \xi}
        +\tilde{M}_Ef\left( x_{i}^{\left( 0\sim m_i-1 \right)}|_{i=1\sim \xi},\zeta ,t \right) 
        +\tilde{M}_EB\left( y,\zeta ,t \right)  u\\
        y=\tilde{C}\tilde{x}_{i}^{\left( 0\sim m_i-1 \right)}|_{i=1\sim \xi},
    \end{array}
    \right. 
\end{align}
where the new variables and matrices are defined by
\begin{align*}
    &\tilde{x}_{i}^{\left( 0\sim m_i-1 \right)}|_{i=1\sim \xi}=\left[ \begin{array}{c}
        x_{i}^{\left( 0\sim m_i-1 \right)}|_{i=1\sim \xi}\\
        d\\
    \end{array} \right]\in \mathbb{R}^{s+q},
    ~E=\left[ \begin{matrix}
        I_s&		0\\
        0&		0\\
    \end{matrix} \right]\in \mathbb{R}^{\left( s+q \right) \times \left( s+q \right)},\\
    &\tilde{P}=\left[ \begin{matrix}
        \Phi _E(0)&		M_ED_1(y)\\
        0&		0\\
    \end{matrix} \right]\in \mathbb{R}^{\left( s+q \right) \times \left( s+q \right)},
    ~\tilde{C}=\left[ \begin{matrix}
        C&		D_2\\
    \end{matrix} \right]\in \mathbb{R}^{p\times \left( s+q \right)},
    ~\tilde{M}_E=[M_{E}^{\mathrm{T}} ~0_{q\times r}^{\mathrm{T}} ] ^{\mathrm{T}}.
\end{align*}

The observability of the constructed descriptor system \eqref{eq:descriptor FAS} is consistent with the original system \eqref{eq:state-space form of original FAS}, which can be easily proven via Assumption \ref{asp:observability and d's dimension constraint}.
By introducing an equality constraint, the descriptor state-space form of FAS \eqref{eq:descriptor FAS} can be equivalently transformed into
\begin{equation}	\label{eq:transformed FAS with TE+NC=I constraint}
    \begin{aligned}
        \begin{array}{l}
            \dot{\tilde{x}}_{i}^{\left( 0\sim m_i-1 \right)}|_{i=1\sim \xi}
            =T\tilde{P}\tilde{x}_{i}^{\left( 0\sim m_i-1 \right)}|_{i=1\sim \xi}
            +T\tilde{M}_{E}f\left( x_{i}^{\left( 0\sim m_i-1 \right)}|_{i=1\sim \xi},\zeta ,t \right) 
            +T\tilde{M}_{E}B\left( y,\zeta ,t \right)  u+N\dot{y}\\
            \textrm{subject to}~~~~TE+N\tilde{C}=I,
        \end{array}
    \end{aligned}
\end{equation}
where $T\in \mathbb{R}^{\left( s+q \right) \times \left( s+q \right)}$ and $N\in \mathbb{R}^{\left( s+q \right) \times p}$ are two additional design parameters to be determined later.

The following lemmas will be instrumental in the subsequent design process.

\begin{lemma}   \label{lemma:general solution to matrix equation}
    \cite{mono-Generalized_inverses}
    Let $\mathcal{A}$, $\mathcal{B}$, and $\mathcal{Y}$ be matrices of appropriate dimensions. The matrix equation
    \[
    \mathcal{A}\mathcal{X}\mathcal{B} = \mathcal{Y}
    \]
    is consistent if and only if there exist matrices $\mathcal{A}^{\dagger}$ and $\mathcal{B}^{\dagger}$ such that
    \[
    \mathcal{A}\mathcal{A}^{\dagger}\mathcal{X}\mathcal{B}^{\dagger}\mathcal{B} = \mathcal{Y}.
    \]
    The general form of the solution is expressed by
    \[
    \mathcal{X} = \mathcal{A}^{\dagger}\mathcal{Y}\mathcal{B}^{\dagger} + \mathcal{S} - \mathcal{A}^{\dagger}\mathcal{A}\mathcal{S}\mathcal{B}\mathcal{B}^{\dagger},
    \]
    where $\mathcal{S}$ is an arbitrary matrix of compatible dimension and the superscript $\dagger$ refers to the Moore-Penrose pseudoinverse, defined as $\mathcal{A}^{\dagger} = (\mathcal{A}^{\mathrm{T}}\mathcal{A})^{-1}\mathcal{A}^{\mathrm{T}}$.
\end{lemma}

\begin{lemma}   \label{lemma:Lyapunov inequality with given decay rate}
    \cite{Duan_IJSS_III_Robust}
    Let $A \in \mathbb{R}^{n \times n}$ be a matrix such that
    \[
    \operatorname{Re} \lambda_i(A) \le -\mu, ~i=1,2,\ldots,n,
    \]
    where $\mu$ is a positive scalar. Then, there exists a positive definite matrix $P \in \mathbb{R}^{n \times n}$ that satisfies
    \[
    A^{\mathrm{T}}P + PA \le -2\mu P.
    \]
\end{lemma}

\begin{lemma}   \label{lemma:parametric design approach}
    \cite{Duan_IJSS_VII_Controllability}
    Let $i \in \{1,2,\ldots,\xi\}$. For any arbitrarily chosen matrix $F_i \in \mathbb{R}^{m_ir_i \times m_ir_i}$, the matrices $A_{0\sim m_i-1}^{i}$ and $V_i \in \mathbb{R}^{m_ir_i \times m_ir_i}$ that fulfill the condition $\det(V_i) \ne 0$ and
    \[
    \Phi_i\left( A_{0\sim m_i-1}^{i}\right)  = V_iF_iV_i^{-1}
    \]
    are determined by
    \[
    A_{0\sim m_i-1}^{i} = -Z_iF_i^{m_i}V_i^{-1}\left(Z_i, F_i \right), 
    \]
    and
    \[
    V_i = V_i\left(Z_i, F_i \right) = 
    \begin{bmatrix}
        Z_i\\
        Z_iF_i\\
        \vdots \\
        Z_iF_i^{m_i-1}
    \end{bmatrix}
    \]
    where $Z_i \in \mathbb{R}^{r_i \times m_ir_i}$ is a parameter matrix that must satisfy
    \[ \det V_i\left(Z_i, F_i \right) \ne 0.  \]
\end{lemma}

\begin{lemma}   \label{lemma:generalized square inequality}
    \cite{Mono-Duan_LMIs}
    The inequality
    \[
    2x^{\operatorname{T}}Py \le \delta x^{\operatorname{T}}Px + \frac{1}{\delta}y^{\operatorname{T}}Py
    \]
    holds, where vector $x,y \in \mathbb{R}^n$, scalar $\delta > 0$, and matrix $P \in \mathbb{R}^{n \times n} > 0$.
\end{lemma}

\section{Active fault/disturbance compensation control}	\label{sec:Active fault/disturbance compensated control}

\subsection{Nonlinear observer design for FAS}

For the transformed system described by \eqref{eq:transformed FAS with TE+NC=I constraint}, we propose the following nonlinear observer for FAS \eqref{eq:compact FAS model}, which facilitates the simultaneous estimation of both the system state and faults/disturbances:
\begin{equation}    \label{eq:oberver}
    \left\{ \begin{array}{l}
        \begin{array}{ll}
            \dot{\varsigma}=&T\tilde{P}\hat{\tilde{x}}_{i}^{\left( 0\sim m_i-1 \right)}|_{i=1\sim \xi}
            +T\tilde{M}_{E}f\left( \hat{x}_{i}^{\left( 0\sim m_i-1 \right)}|_{i=1\sim \xi},\zeta ,t \right) \\
            &+T\tilde{M}_{E}B\left( y,\zeta ,t \right) u+L\left( y-\tilde{C}\hat{\tilde{x}}_{i}^{\left( 0\sim m_i-1 \right)}|_{i=1\sim \xi} \right)\\
        \end{array}\\
        \hat{\tilde{x}}_{i}^{\left( 0\sim m_i-1 \right)}|_{i=1\sim \xi}=\varsigma +Ny\\
        \hat{x}_{i}^{\left( 0\sim m_i-1 \right)}|_{i=1\sim \xi}=H_1\hat{\tilde{x}}_{i}^{\left( 0\sim m_i-1 \right)}|_{i=1\sim \xi}\\
        \hat{d} = H_2 \hat{\tilde{x}}_{i}^{\left( 0\sim m_i-1 \right)}|_{i=1\sim \xi},
    \end{array} \right.
\end{equation}
where $\varsigma$ is an intermediate variable, $\hat{\tilde{x}}_{i}^{\left( 0\sim m_i-1 \right)}|_{i=1\sim \xi}$, $\hat{x}_{i}^{\left( 0\sim m_i-1 \right)}|_{i=1\sim \xi}$, and $\hat{d}$ correspondingly represent the estimated values.
The nonlinear observer \eqref{eq:oberver} operates with three key parameters: in addition to the previously introduced matrices $T$ and $N$ in \eqref{eq:transformed FAS with TE+NC=I constraint}, the observer gain matrix $L\in \mathbb{R}^{\left( s+q \right) \times p}$ will also be specified later. The decoupling matrices $H_1$ and $H_2$ are given by
\[
H_1=\left[ \begin{matrix}
    I_s&		0_{s\times q}\\
\end{matrix} \right], ~
H_2=\left[ \begin{matrix}
    0_{q\times s}&		I_q\\
\end{matrix} \right].
\]
The last two equations in the proposed observer \eqref{eq:oberver} are responsible for reconstructing the system state $\hat{x}_{i}^{\left( 0\sim m_i-1 \right)}|_{i=1\sim \xi}$ and the fault/disturbance signal $d$, respectively. Similarly, corresponding relationships are also present in \eqref{eq:descriptor FAS} where they are defined as follows:
\begin{align*}
    x_{i}^{\left( 0\sim m_i-1 \right)}|_{i=1\sim \xi}=H_1\tilde{x}_{i}^{\left( 0\sim m_i-1 \right)}|_{i=1\sim \xi},
    ~d=H_2\tilde{x}_{i}^{\left( 0\sim m_i-1 \right)}|_{i=1\sim \xi}.
\end{align*}

Define the estimation error as
\begin{align}	\label{eq:estimation error}
    e=\tilde{x}_{i}^{\left( 0\sim m_i-1 \right)}|_{i=1\sim \xi}-\hat{\tilde{x}}_{i}^{\left( 0\sim m_i-1 \right)}|_{i=1\sim \xi}.
\end{align}
Taking the derivative of the estimation error yields
\begin{align*}
    \dot{e}
    =&\dot{\tilde{x}}_{i}^{\left( 0\sim m_i-1 \right)}|_{i=1\sim \xi}-\dot{\hat{\tilde{x}}}_{i}^{\left( 0\sim m_i-1 \right)}|_{i=1\sim \xi}\\
    =&T\tilde{P}\tilde{x}_{i}^{\left( 0\sim m_i-1 \right)}|_{i=1\sim \xi}+T\tilde{M}_{E}f\left( x_{i}^{\left( 0\sim m_i-1 \right)}|_{i=1\sim \xi},\zeta ,t \right) \\
    &+T\tilde{M}_{E}B\left( y,\zeta ,t \right)  u -T\tilde{P}\hat{\tilde{x}}_{i}^{\left( 0\sim m_i-1 \right)}|_{i=1\sim \xi}
    -T\tilde{M}_{E}f\left( \hat{x}_{i}^{\left( 0\sim m_i-1 \right)}|_{i=1\sim \xi},\zeta ,t \right) \\
    &-T\tilde{M}_{E}B\left( y,\zeta ,t \right) u-L\left( y-\tilde{C}\hat{\tilde{x}}_{i}^{\left( 0\sim m_i-1 \right)}|_{i=1\sim \xi} \right),
\end{align*}
i.e., 
\begin{align}   \label{eq:observer error dynamics}
    \dot{e}=\left( T\tilde{P}-L\tilde{C} \right) e+T\tilde{M}_{E}\varDelta f,
\end{align}
\[
\varDelta f=f\left( x_{i}^{\left( 0\sim m_i-1 \right)}|_{i=1\sim \xi},\zeta ,t \right) -f\left( \hat{x}_{i}^{\left( 0\sim m_i-1 \right)}|_{i=1\sim \xi},\zeta ,t \right).
\]

Unlike conventional observers, the proposed observer requires not only a suitable gain matrix $L$ to stabilize the error system \eqref{eq:observer error dynamics}, but also must satisfy the equality constraint specified in \eqref{eq:transformed FAS with TE+NC=I constraint}. Applying standard observer gain design methods directly could violate this equality constraint and lead to divergent error dynamics. To address this, we utilize generalized matrix inverse techniques as outlined in Lemma \ref{lemma:general solution to matrix equation}. Consequently, the observer parameters $T$ and $N$ can be determined by
\begin{align}
    &T=\varTheta ^{\dag}H_3+SH_3-S\varTheta \varTheta ^{\dag}H_3,  \label{eq:T}\\    
    &N=\varTheta ^{\dag}H_4+SH_4-S\varTheta \varTheta ^{\dag}H_4,  \label{eq:N}
\end{align}
where $S\in \mathbb{R}^{\left( s+q \right) \times \left( s+q+p \right)}$ is a free matrix, and
\[
\varTheta =\left[ \begin{array}{c}
    E\\
    \tilde{C}\\
\end{array} \right] ,H_3=\left[ \begin{array}{c}
    I_{s+q}\\
    0_{p\times \left( s+q \right)}\\
\end{array} \right] ,H_4=\left[ \begin{array}{c}
    0_{\left( s+q \right) \times p}\\
    I_p\\
\end{array} \right].
\]

\begin{assumption}  \label{asp:norm bound of D1}
    For the sake of stability analysis, it is assumed that $\left\| D_1\left( y \right) \right\| \le \bar{D}_1$.   
\end{assumption}

Hence, we introduce the following theorem to design observer parameters and ensure asymptotic stability of the error system \eqref{eq:observer error dynamics}.

\begin{theorem}     \label{theorem:observer}
    Suppose there exist a positive definite matrix $P_e \in \mathbb{R}^{(s+q)\times (s+q)}$, full matrices $Q \in \mathbb{R}^{(s+q) \times p}$ and $W \in \mathbb{R}^{(s+q) \times (s+q+p)}$, a scalar variable $\eta > 0$, and a prescribed scalar $\mu_e > 0$, such that the following matrix inequality constraint holds:
    \begin{align}  \label{eq:observer LMI}
        \left[ \begin{matrix}
            \varLambda _{11}&		\varLambda _{12}\\
            \ast&		-\eta I\\
        \end{matrix} \right] <0,
    \end{align}
    \begin{align*}
        \varLambda _{11}=&\mathbf{sym}\left( P_e\varTheta ^{\dag}H_3\tilde{P}+WH_3\tilde{P}-W\varTheta \varTheta ^{\dag}H_3\tilde{P} \right) 
        -\mathbf{sym}\left( Q\tilde{C} \right) +2\mu _eP_e+\eta \gamma _{f}^{2}H_{1}^{\mathrm{T}}H_1,\\
        \varLambda _{12}=&P_e\varTheta ^{\dag}H_3\tilde{M}_E+WH_3\tilde{M}_E-W\varTheta \varTheta ^{\dag}H_3\tilde{M}_E,
    \end{align*}
    the poles of the proposed observer \eqref{eq:oberver} will be placed at $\operatorname{Re}(\lambda_{i}) < -\mu_e, i=1,2,\ldots,s+q$, and the error dynamics \eqref{eq:observer error dynamics} is asymptotically stable as guaranteed by the following exponential decay rate
    \begin{align}  \label{eq:observer error performance}
        \lVert e \rVert ^2\le \lambda _{\min}^{-1}\left( P_e \right) V_e\left( 0 \right) \exp(-c_1t),
    \end{align}
    where $c_1=\mu _e-\mu _{e}^{-1}\gamma _{f}^{2}\lambda _{\min}^{-1}\left( P_e \right) \lVert P_e \rVert \lVert T\tilde{M}_E \rVert ^2\lVert H_1 \rVert ^2$ and $V_e(0)$ represents the initial value of the selected Lyapunov function.
    The observer parameters $T$, $N$, and $L$ can be determined by
    \begin{align}   \label{eq:T N L}
        \begin{aligned}
            &L=P_{e}^{-1}Q, ~S=P_{e}^{-1}W, \\
            &T=\varTheta ^{\dag}H_3+SH_3-S\varTheta \varTheta ^{\dag}H_3, \\
            &N=\varTheta ^{\dag}H_4+SH_4-S\varTheta \varTheta ^{\dag}H_4.
        \end{aligned}
    \end{align}
\end{theorem}

\begin{proof}
    The proof is divided into two parts. First, we demonstrate that the design condition \eqref{eq:observer LMI} is satisfied, and then we analyze the estimation performance.
    
    Assuming the inequality constraint \eqref{eq:observer LMI} is met, and by substituting $Q=P_eL$ and $W=P_eS$, along with the expressions for $T$ and $N$ provided in \eqref{eq:T} and \eqref{eq:N}, respectively, we obtain
    \begin{align*}
        \left[ \begin{matrix}
            \varLambda_{11}^{\prime}&		P_eT\tilde{M}_E\\
            \ast&		-\eta I\\
        \end{matrix} \right] < 0,
    \end{align*}
    \[
    \varLambda_{11}^{\prime} = \mathbf{sym}\left( P_e\left( T\tilde{P}-L\tilde{C} \right)\right)  +2\mu _eP_e+\eta \gamma _{f}^{2}H_{1}^{\mathrm{T}}H_1.
    \]
    Pre- and post-multiplying $\left[e^{\mathrm{T}} ~~\Delta f^{\mathrm{T}} \right]$ and $\left[e^{\mathrm{T}} ~~\Delta f^{\mathrm{T}} \right]^{\mathrm{T}}$, respectively, yields
    \begin{align*}
        e^{\mathrm{T}}\mathbf{sym}\left( P_e\left( T\tilde{P}-L\tilde{C} \right) \right) e+2\mu _eP_e
        +2e^{\mathrm{T}}P_eT\tilde{M}_E\varDelta f+\eta \gamma _{f}^{2}e^{\mathrm{T}}H_{1}^{\mathrm{T}}H_1e-\eta \varDelta f^{\mathrm{T}}\varDelta f < 0.
    \end{align*}
    
    According to Assumption \ref{asp:Lipshitz condition}, we have the inequality relationship $\lVert \varDelta f \rVert \le \gamma _f\lVert H_1\tilde{x}_{i}^{\left( 0\sim m_i-1 \right)}|_{i=1\sim \xi}-H_1\hat{\tilde{x}}_{i}^{\left( 0\sim m_i-1 \right)}|_{i=1\sim \xi} \rVert$, which can be further written as
    \[
    \lVert \varDelta f \rVert \le \gamma _f\lVert H_1e \rVert \le \gamma _f\lVert H_1 \rVert \lVert e \rVert,
    \]
    indicating that $\eta \gamma _{f}^{2}e^{\mathrm{T}}H_{1}^{\mathrm{T}}H_1e-\eta \varDelta f^{\mathrm{T}}\varDelta f \ge 0$, with a scalar coefficient $\eta > 0$. Besides, we know that $2\mu _eP_e > 0$. Then, we have
    \[
    e^{\mathrm{T}} \mathbf{sym}\left(P_e\left( T\tilde{P}-L\tilde{C} \right) \right) e+2e^{\mathrm{T}}P_eT\tilde{M}_E\varDelta f < 0.
    \]
    By taking the Lyapunov function $V_e=e^{\mathrm{T}}P_ee$, we arrive at
    \begin{align*}
        \dot{V}_e&=\dot{e}^{\mathrm{T}}P_ee+e^{\mathrm{T}}P_e\dot{e} < 0.
    \end{align*}
    Thus, the error system in \eqref{eq:observer error dynamics} is proven to be stable.
    
    Given that the design condition in \eqref{eq:observer LMI} holds, analyzing the Lyapunov function yields
    \begin{align*}
        \dot{V}_e
        =\dot{e}^{\mathrm{T}}P_ee+e^{\mathrm{T}}P_e\dot{e}
        =e^{\mathrm{T}}\mathbf{sym}\left( P_e\left( T\tilde{P}-L\tilde{C} \right) \right) e+2e^{\mathrm{T}}P_eT\tilde{M}_E\varDelta f.
    \end{align*}
    Based on Lemma \ref{lemma:Lyapunov inequality with given decay rate}, we obtain
    \[
    \dot{V}_e \le -2\mu _eV_e+2e^{\mathrm{T}}P_eT\tilde{M}_E\varDelta f.
    \]
    Using Lemma \ref{lemma:generalized square inequality}, it comes to
    \begin{align*}
        \dot{V}_e
        &\le -2\mu _eV_e+\mu _ee^{\mathrm{T}}P_ee+\mu _{e}^{-1}\left( T\tilde{M}_E\varDelta f \right) ^{\mathrm{T}}P_eT\tilde{M}_E\varDelta f\\
        &\le -\mu _eV_e+\mu _{e}^{-1}\lVert P_e \rVert \lVert T\tilde{M}_E \rVert ^2\lVert \varDelta f \rVert ^2.
    \end{align*}
    Hence,
    \begin{align*}
        \dot{V}_e
        &\le -\mu _eV_e+\mu _{e}^{-1}\gamma _{f}^{2}\lVert P_e \rVert \lVert T\tilde{M}_E \rVert ^2\lVert H_1 \rVert ^2\lVert e \rVert ^2\\
        &\le -\mu _eV_e+\mu _{e}^{-1}\gamma _{f}^{2}\lVert P_e \rVert \lVert T\tilde{M}_E \rVert ^2\lVert H_1 \rVert ^2\lambda _{\min}^{-1}\left( P_e \right) V_e.
    \end{align*}
    
    We have the following ordinary differential inequality, 
    \[
    \dot{V}_e\le -\left( \mu _e-\mu _{e}^{-1}\gamma _{f}^{2}\lambda _{\min}^{-1}\left( P_e \right) \lVert P_e \rVert \lVert T\tilde{M}_E \rVert ^2\lVert H_1 \rVert ^2 \right) V_e,
    \]
    and according to the comparison theorem \cite{mono-Ames_InequalitiesForDifferential}, the solution is given by
    \[
    V_e\left( t \right) \le V_e\left( 0 \right) \exp(-c_1t),
    \]
    where $V_e(0)$ is the initial value of Lyapunov function and $c_1=\mu _e-\mu _{e}^{-1}\gamma _{f}^{2}\lambda _{\min}^{-1}\left( P_e \right) \lVert P_e \rVert \lVert T\tilde{M}_E \rVert ^2\lVert H_1 \rVert ^2$. 
    
    Finally, the estimation performance \eqref{eq:observer error performance} is obtained, indicating that the error vector will asymptotically converge to the origin.
\end{proof}

\subsection{Observer-based control with active compensation}

Taking advantage of the estimation results from the proposed observer \eqref{eq:oberver} and inspired by the general FAS controller structure \cite{Duan_IJSS_I_Models,Duan_IJSS_II_strict-feedback,Duan_IJSS_V_RobustAdaptive,Duan_IJSS_IV_Adaptive}, the control law with active fault/disturbance compensation is formulated as
\begin{align}   \label{eq:controller}
    \left\{ \begin{array}{l}
        u=-B^{-1}\left( y,\zeta ,t \right) \left( K\hat{x}_{i}^{\left( 0\sim m_i-1 \right)}|_{i=1\sim \xi}+u^{\ast} \right)\\
        u^{\ast}=f\left( \hat{x}_{i}^{\left( 0\sim m_i-1 \right)}|_{i=1\sim \xi},\zeta ,t \right) +D_1(y)\hat{d},
    \end{array} \right. 
\end{align}
where $K=\operatorname{blockdiag}\left( A_{0\sim m_i-1}^{i},i=1,2,...,\xi \right)$. The individual controller gain matrix, i.e., $A_{0\sim m_i-1}^{i}$, can be determined using the parametric design approach outlined in Lemma \ref{lemma:parametric design approach}, which ultimately yields a linear, constant, closed-loop control system with prescribed pole locations. The term $\hat{d}$ serves as a compensation element to counteract fault/disturbance in the actuators. Consequently, the combination of the proposed observer \eqref{eq:oberver} and control law \eqref{eq:controller} ensures the asymptotic convergence of the system, as stated in the following theorem.

\begin{theorem}		\label{theorem:controller}
    The FAS, described by either \eqref{eq:compact FAS model} or \eqref{eq:big FAS model}, is asymptotically stable under the proposed observer-based controller \eqref{eq:controller}. Meanwhile, the state of the closed-loop control system exhibits the subsequent performance characteristics:
    \begin{align}      \label{eq:control performance}
        \lVert x_{i}^{\left( 0\sim m_i-1 \right)}|_{i=1\sim \xi} \rVert ^2
        \le \lambda _{\min}^{-1}\left( P_x \right) \left( V_x\left( 0 \right) -\frac{c_3}{\mu _x-c_1} \right) \exp(-\mu _xt)
        +\lambda _{\min}^{-1}\left( P_x \right) \frac{c_3}{\mu _x-c_1}\exp(-c_1t),
    \end{align}
    where $c_2=\mu _{x}^{-1}\lVert P_x \rVert \big( \lVert \Phi _E\left( A \right) H_1 \rVert ^2+\gamma _{f}^{2}\lVert M_E \rVert ^2\lVert H_1 \rVert ^2+\lVert M_E\bar{D}_1H_2 \rVert ^2 \big)$, $c_3=c_2\lambda _{\min}^{-1}\left( P_e \right) V_e\left( 0 \right)$, and $V_x(0)$ represents the value of the chosen Lyapunov function at the starting time.
\end{theorem}

\begin{proof}
    Based on the proposed controller \eqref{eq:controller}, the closed-loop system can be obtained by
    \begin{align*}
        x_{i}^{\left( m_i \right)}|_{i=1\sim \xi}
        =&-K\hat{x}_{i}^{\left( 0\sim m_i-1 \right)}|_{i=1\sim \xi}+\varDelta f+D_1(y)\left( d-\hat{d} \right)\\
        =&-K\hat{x}_{i}^{\left( 0\sim m_i-1 \right)}|_{i=1\sim \xi}+\varDelta f+D_1(y)H_2
        \left( \tilde{x}_{i}^{\left( 0\sim m_i-1 \right)}|_{i=1\sim \xi}-\hat{\tilde{x}}_{i}^{\left( 0\sim m_i-1 \right)}|_{i=1\sim \xi} \right)\\
        =&-K\hat{x}_{i}^{\left( 0\sim m_i-1 \right)}|_{i=1\sim \xi}+\varDelta f+D_1(y)H_2e.
    \end{align*}
    
    The state-space representation of the closed-loop system is
    \begin{align*}
        \dot{x}_{i}^{\left( 0\sim m_i-1 \right)}
        =\Phi \left( A_{0\sim m_i-1}^{i} \right) \hat{x}_{i}^{\left( 0\sim m_i-1 \right)}+M_{ir}\varDelta f_i
        +M_{ir}D_1(y)H_2e, ~i=1,2,\ldots,\xi,
    \end{align*}
    which can be rewritten into a compact form
    \begin{align*}
        \dot{x}_{i}^{\left( 0\sim m_i-1 \right)}|_{i=1\sim \xi}
        =\Phi _E\left( A \right) \hat{x}_{i}^{\left( 0\sim m_i-1 \right)}|_{i=1\sim \xi}
        +\tilde{M}_E\varDelta f+\tilde{M}_ED_1(y)H_2e.
    \end{align*}
    
    Taking the Lyapunov function 
    \[
    V_x=\left( x_{i}^{\left( 0\sim m_i-1 \right)}|_{i=1\sim \xi} \right) ^{\mathrm{T}}P_xx_{i}^{\left( 0\sim m_i-1 \right)}|_{i=1\sim \xi},
    \]
    where $P_x$ is a positive definite matrix obtained from Lemma \ref{lemma:Lyapunov inequality with given decay rate}. The derivative of the selected Lyapunov function can be written as the sum of two terms, that is
    \[
    \dot{V}_x=\dot{V}_{x}^{\text{I}}+\dot{V}_{x}^{\text{II}},
    \]
    where
    \begin{align*}
        \dot{V}_{x}^{\text{I}}
        =&\left( \hat{x}_{i}^{\left( 0\sim m_i-1 \right)}|_{i=1\sim \xi} \right) ^{\mathrm{T}}\Phi _E\left( A \right) ^{\mathrm{T}}P_xx_{i}^{\left( 0\sim m_i-1 \right)}|_{i=1\sim \xi} 
        +\left( x_{i}^{\left( 0\sim m_i-1 \right)}|_{i=1\sim \xi} \right) ^{\mathrm{T}}P_x\Phi _E\left( A \right) \hat{x}_{i}^{\left( 0\sim m_i-1 \right)}|_{i=1\sim \xi}, \\
        \dot{V}_{x}^{\text{II}}
        =&2\left( x_{i}^{\left( 0\sim m_i-1 \right)}|_{i=1\sim \xi} \right) ^{\mathrm{T}}P_x\tilde{M}_E\varDelta f 
        +2\left( x_{i}^{\left( 0\sim m_i-1 \right)}|_{i=1\sim \xi} \right) ^{\mathrm{T}}P_x\tilde{M}_ED_1(y)H_2e.
    \end{align*}
    
    Considering the estimation error defined in \eqref{eq:estimation error} and construction equation of $\hat{x}_{i}^{\left( 0\sim m_i-1 \right)}|_{i=1\sim \xi}$ in \eqref{eq:oberver}, the following relationship holds
    \begin{align}	\label{eq:overall relation xhat with x and e}
        \hat{x}_{i}^{\left( 0\sim m_i-1 \right)}|_{i=1\sim \xi}=x_{i}^{\left( 0\sim m_i-1 \right)}|_{i=1\sim \xi}-H_1e.
    \end{align}
    We then derive the first part in $\dot{V}_x$:
    \begin{align*}
        \dot{V}_{x}^{\text{I}}
        =&\left( H_1\hat{\tilde{x}}_{i}^{\left( 0\sim m_i-1 \right)}|_{i=1\sim \xi} \right) ^{\mathrm{T}}\Phi _E\left( A \right) ^{\mathrm{T}}P_xx_{i}^{\left( 0\sim m_i-1 \right)}|_{i=1\sim \xi}\\
        &+\left( x_{i}^{\left( 0\sim m_i-1 \right)}|_{i=1\sim \xi} \right) ^{\mathrm{T}}P_x\Phi _E\left( A \right) H_1\hat{\tilde{x}}_{i}^{\left( 0\sim m_i-1 \right)}|_{i=1\sim \xi}\\
        =&\left( H_1\left( \tilde{x}_{i}^{\left( 0\sim m_i-1 \right)}|_{i=1\sim \xi}-e \right) \right) ^{\mathrm{T}}\Phi _E\left( A \right) ^{\mathrm{T}}
         P_xx_{i}^{\left( 0\sim m_i-1 \right)}|_{i=1\sim \xi}\\
        &+\left( x_{i}^{\left( 0\sim m_i-1 \right)}|_{i=1\sim \xi} \right) ^{\mathrm{T}}P_x 
         \Phi _E\left( A \right) H_1\left( \tilde{x}_{i}^{\left( 0\sim m_i-1 \right)}|_{i=1\sim \xi}-e \right)\\
        =&\left( x_{i}^{\left( 0\sim m_i-1 \right)}|_{i=1\sim \xi}-H_1e \right) ^{\mathrm{T}}\Phi _E\left( A \right) ^{\mathrm{T}}
         P_xx_{i}^{\left( 0\sim m_i-1 \right)}|_{i=1\sim \xi}\\
        &+\left( x_{i}^{\left( 0\sim m_i-1 \right)}|_{i=1\sim \xi} \right) ^{\mathrm{T}}
         P_x\Phi _E\left( A \right) \left( x_{i}^{\left( 0\sim m_i-1 \right)}|_{i=1\sim \xi}-H_1e \right)\\
        =&\left( x_{i}^{\left( 0\sim m_i-1 \right)}|_{i=1\sim \xi} \right) ^{\mathrm{T}}\mathbf{sym}\left( P_x\Phi _E\left( A \right) \right) x_{i}^{\left( 0\sim m_i-1 \right)}|_{i=1\sim \xi}
        -2\left( x_{i}^{\left( 0\sim m_i-1 \right)}|_{i=1\sim \xi} \right) ^{\mathrm{T}}P_x\Phi _E\left( A \right) H_1e.
    \end{align*}
    In light of Lemmas \ref{lemma:Lyapunov inequality with given decay rate} and \ref{lemma:generalized square inequality}, the upper bound of $\dot{V}_{x}^{\text{I}}$ can be derived by
    \begin{align*}
        \dot{V}_{x}^{\text{I}}
        \le& -2\mu _xV_x-2\left( x_{i}^{\left( 0\sim m_i-1 \right)}|_{i=1\sim \xi} \right) ^{\mathrm{T}}P_x\Phi _E\left( A \right) H_1e\\
        \le& -2\mu _xV_x-\mu _x\left( x_{i}^{\left( 0\sim m_i-1 \right)}|_{i=1\sim \xi} \right) ^{\mathrm{T}}P_xx_{i}^{\left( 0\sim m_i-1 \right)}|_{i=1\sim \xi}
        -\mu _{x}^{-1}\left( \Phi _E\left( A \right) H_1e \right) ^{\mathrm{T}}P_x\Phi _E\left( A \right) H_1e\\
        \le& -3\mu _xV_x-\mu _{x}^{-1}\left( \Phi _E\left( A \right) H_1e \right) ^{\mathrm{T}}P_x\Phi _E\left( A \right) H_1e.
    \end{align*}
    Then, some necessary derivations should be conducted on $\dot{V}_{x}^{\text{II}}$.
    \begin{align*}
        \dot{V}_{x}^{\text{II}}
        \le& \mu _x\left( x_{i}^{\left( 0\sim m_i-1 \right)}|_{i=1\sim \xi} \right) ^{\mathrm{T}}P_xx_{i}^{\left( 0\sim m_i-1 \right)}|_{i=1\sim \xi}
        +\mu _{x}^{-1}\left( \tilde{M}_E\varDelta f \right) ^{\mathrm{T}}P_x\tilde{M}_E\varDelta f\\
        &+\mu _x\left( x_{i}^{\left( 0\sim m_i-1 \right)}|_{i=1\sim \xi} \right) ^{\mathrm{T}}P_xx_{i}^{\left( 0\sim m_i-1 \right)}|_{i=1\sim \xi}
        +\mu _{x}^{-1}\left( \tilde{M}_ED_1(y)H_2e \right) ^{\mathrm{T}}P_x\tilde{M}_ED_1(y)H_2e\\
        =&2\mu _xV_x+\mu _{x}^{-1}\left( \tilde{M}_E\varDelta f \right) ^{\mathrm{T}}P_x\tilde{M}_E\varDelta f
        +\mu _{x}^{-1}\left( \tilde{M}_ED_1(y)H_2e \right) ^{\mathrm{T}}P_x\tilde{M}_ED_1(y)H_2e.
    \end{align*}
    Add both of $\dot{V}_{x}^{\text{I}}$ and $\dot{V}_{x}^{\text{II}}$ and use Assumption \ref{asp:norm bound of D1}, resulting in
    \begin{align*}
        \dot{V}_x
        =&\dot{V}_{x}^{\text{I}}+\dot{V}_{x}^{\text{II}}\\
        \le& -\mu _xV_x+\mu _{x}^{-1}\lVert \Phi _E\left( A \right) H_1 \rVert ^2\lVert P_x \rVert \lVert e \rVert ^2
        +\mu _{x}^{-1}\gamma _{f}^{2}\lVert P_x \rVert \lVert M_E \rVert ^2\lVert H_1 \rVert ^2\lVert e \rVert ^2
        +\mu _{x}^{-1}\lVert P_x \rVert \lVert M_E\bar{D}_1H_2 \rVert ^2\lVert e \rVert ^2\\
        =&-\mu _xV_x+c_2\lVert e \rVert ^2,
    \end{align*}
    where $c_2=\mu _{x}^{-1}\lVert P_x \rVert \big( \lVert \Phi _E\left( A \right) H_1 \rVert ^2+\gamma _{f}^{2}\lVert M_E \rVert ^2\lVert H_1 \rVert ^2+\lVert M_E\bar{D}_1H_2 \rVert ^2 \big)$. 
    Substituting \eqref{eq:observer error performance} into the above yields
    \begin{align*}
        \dot{V}_x \le -\mu _xV_x+c_3\exp(-c_1t),
    \end{align*}
    where $c_3=c_2\lambda _{\min}^{-1}\left( P_e \right) V_e\left( 0 \right)$. Finally, we get the result shown in \eqref{eq:control performance} and validate that the closed-loop system is asymptotically stable, i.e., $x \to 0$ when $t \to \infty$.
    This completes the proof.
\end{proof}

To enhance clarity, the flow of signals and control actions in the proposed active fault/disturbance compensation control scheme is summarized in Figure \ref{fig:blockdiagram}.
\begin{figure}[htbp]
    \centering
    \includegraphics[scale=0.7]{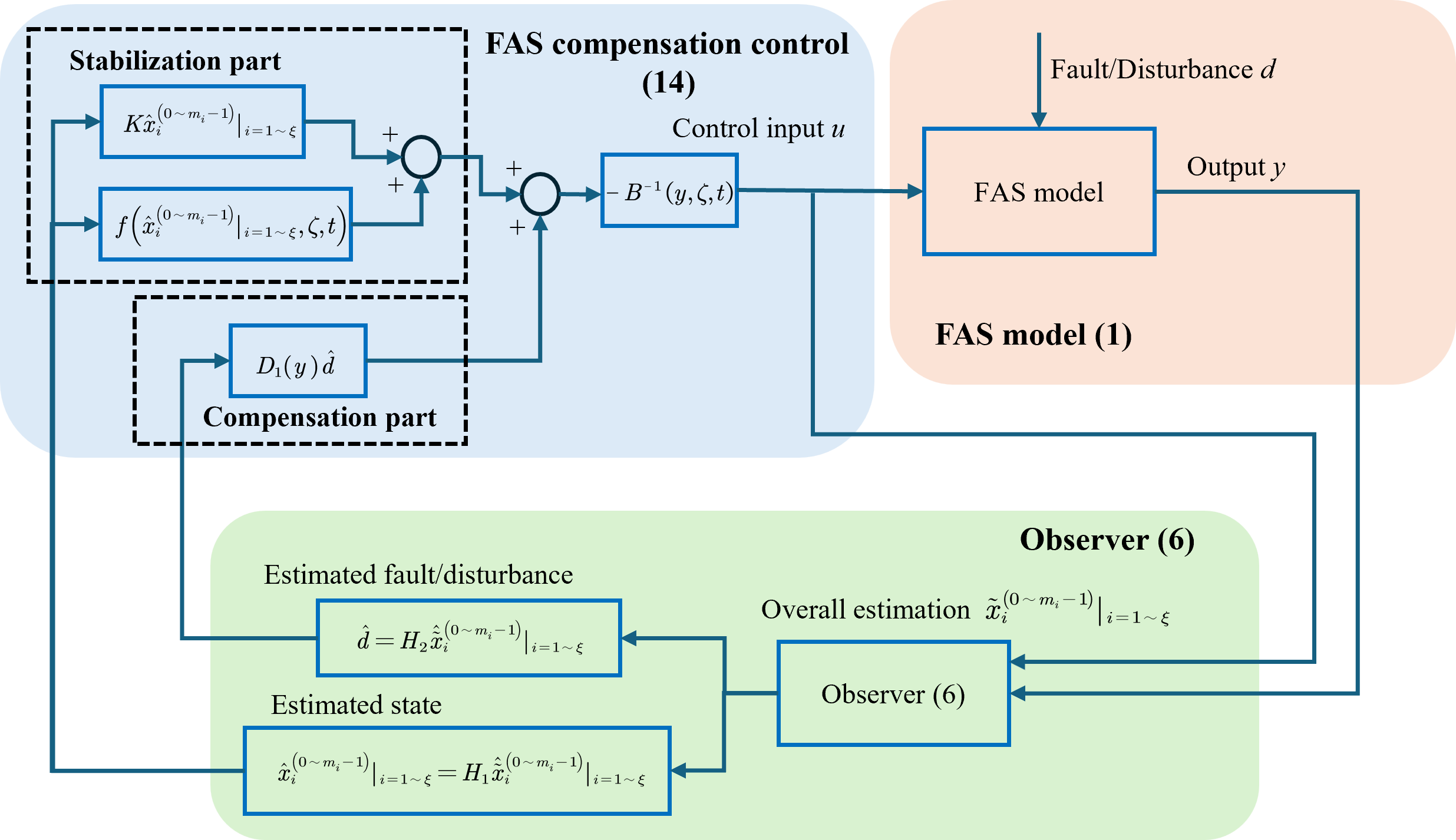}
    \caption{Block diagram of the proposed observer-based fault/disturbance compensation control.}
    \label{fig:blockdiagram}
\end{figure}

\section{Illustrative examples}		\label{sec:Simulation and experiment}

This section demonstrates the feasibility, superiority, and practicality of the proposed method through a simulation using an electromechanical system and an actual experiment with a Quanser Ball and Beam system.

\subsection{Comparative simulation for electromechanical system}

This subsection demonstrates both the feasibility and advantage of the proposed observer-based active fault/disturbance compensation control using an electromechanical system. To provide a more compelling validation, the estimation and control outcomes are compared with those presented in \cite{JiangDuanHou2024TCSI}.

The dynamics of the electromechanical system can be characterized by the following equations:
\begin{equation}	\label{eq:sim_electromechanical system}
    \left\{ \begin{array}{l}
        M_e\ddot{q}+B_e\dot{q}+N_e\sin \left( q \right) =I\\
        L_e\dot{I}+R_eI+K_B\dot{q}=V_e+d,
    \end{array} \right. 
\end{equation}
where
\begin{align*}
    M_e=\frac{J_e}{K_{\tau}}+\frac{m_eL_{o}^{2}}{3K_{\tau}}+\frac{M_oL_{o}^{2}}{K_{\tau}}+\frac{2M_oR_{o}^{2}}{5K_{\tau}},~~
    N_e=\frac{m_eL_og}{2K_{\tau}}+\frac{M_oL_og}{K_{\tau}}, ~B_e=\frac{B_o}{K_{\tau}}.
\end{align*}

In \eqref{eq:sim_electromechanical system}, $q$ denotes the angular position of the link, $I$ is the motor armature current, and $V_e$ represents the voltage input. The variable $d$ accounts for faults or disturbances. The system parameters are described as follows: $L_e = 0.025 ~\mathrm{H}$ is the armature inductance, $R_e = 5.0 ~\mathrm{\Omega}$ is the armature resistance, $K_B = 0.90 ~\mathrm{N\cdot m/A}$ is the back-emf coefficient, $J_e = 1.625\times 10^{-3} ~\mathrm{kg\cdot m^2}$ is the rotor inertia, $m_e = 0.506 ~\mathrm{kg}$ is the link mass, $M_o = 0.434 ~\mathrm{kg}$ is the load mass, $L_o = 0.305 ~\mathrm{m}$ is the link length, $R_o = 0.023 ~\mathrm{m}$ is the radius of load, $g = 9.8 ~\mathrm{m/s^2}$ is the gravitational constant, $B_o = 16.25\times 10^{-3} ~\mathrm{N\cdot m\cdot s/rad}$ is the coefficient of viscous friction at the joint, $K_{\tau} = 0.90 ~\mathrm{N\cdot m/A}$ is the conversion coefficient of armature current to torque. 
For details on transforming dynamics \eqref{eq:sim_electromechanical system} into the form of \eqref{eq:compact FAS model} using the FAS approach, one can refer to the foundational work in \cite{Duan_IJSS_I_Models,Duan_IJSS_II_strict-feedback}, or the example section of \cite{JiangDuanHou2024TCSI}. The final result yields the following FAS model:
\begin{align}	\label{eq:sim_FAS model}
    \left\{ \begin{array}{l}
        \dddot{q}=f\left( q^{\left( 0\sim 2 \right)} \right) +\frac{1}{M_eL_e}V_e+D_1(y)d\\
        y=Cq^{\left( 0\sim 2 \right)}+D_2d,\\
    \end{array} \right. 
\end{align}
where
\begin{align*}
    f\left( x^{\left( 0\sim 2 \right)} \right) =-\frac{B_eL_e+M_eR_e}{M_eL_e}\ddot{q}-\frac{R_eB_e+K_B}{M_eL_e}\dot{q}
    -\frac{N_e}{M_e}\dot{q}\cos \left( q \right) -\frac{R_eN_e}{M_eL_e}\sin \left( q \right),
\end{align*}
\[
C=\left[ \begin{matrix}
    1&		0&		0\\
    0&		B_e&		M_e\\
\end{matrix} \right],~
D_1(y) = \frac{1}{M_eL_e},~
D_2 = [0 ~0.1]^{\mathrm{T}}.
\]

According to Theorem \ref{theorem:observer}, setting $\mu_e = 40$ and $\gamma_f = 1$, we solve the linear matrix inequality \eqref{eq:observer LMI} to obtain $\eta = 0.9501$ and the following observer parameters:
\[
T = \begin{bmatrix}
    7.6805 & 0 & 0 & 0 \\
    -42.5471 & 1 & 0 & 0 \\
    -16.3500 & 0 & 1 & 0 \\
    -0.1066 & -0.1806 & -1.6642 & 0
\end{bmatrix},
N = \begin{bmatrix}
    -6.6805 & 0 \\
    42.5471 & 0 \\
    16.3500 & 0 \\
    0.1066 & 10
\end{bmatrix},
L = \begin{bmatrix}
    41.0012 & -1.3840 \\
    7.7106 & -88.9603 \\
    -2.1243 & 1443.4586 \\
    4.0560 & -1994.0520
\end{bmatrix}.
\]

In the following simulation, we set the initial system values to $[q ~I ~\dot{q}]^{\mathrm{T}} = [1 ~1 ~1]^{\mathrm{T}}$ and the initial observer values to $\hat{\tilde{x}}^{(0\sim 2)} = [\hat{q} ~\dot{\hat{q}} ~\ddot{\hat{q}} ~\hat{d}]^{\mathrm{T}} = [6.6805 ~-42.5471 ~-16.3500 ~9.0924]^{\mathrm{T}}$. To demonstrate the feasibility and superiority of the proposed observer, firstly, we evaluate the estimation performance by simulating system \eqref{eq:sim_electromechanical system} with the control input $V_e$ set to identically zero. We compare the proposed observer \eqref{eq:oberver} with PI observer presented in \cite{JiangDuanHou2024TCSI}. The resulting estimation curves for both state and fault are shown in Figure \ref{fig:sim_estimation}. The comparison confirms the feasibility and advantage of the proposed observer, highlighting its accurate state and fault estimation capability and faster convergence speed.

\begin{figure}[htbp]
    \centering
    \begin{subfloat}[]
        {\includegraphics[width=0.4\linewidth]{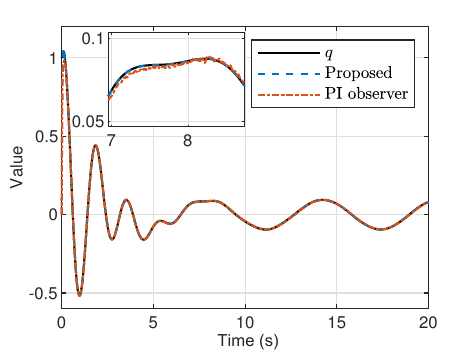}}
    \end{subfloat}
    \begin{subfloat}[]
        {\includegraphics[width=0.4\linewidth]{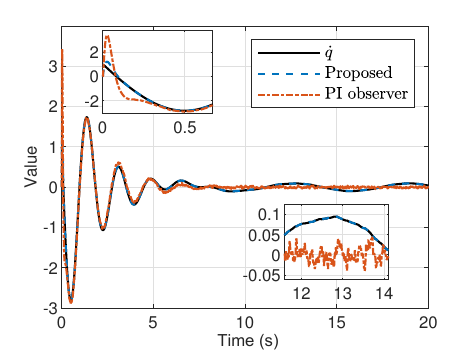}}
    \end{subfloat}
    \begin{subfloat}[]
        {\includegraphics[width=0.4\linewidth]{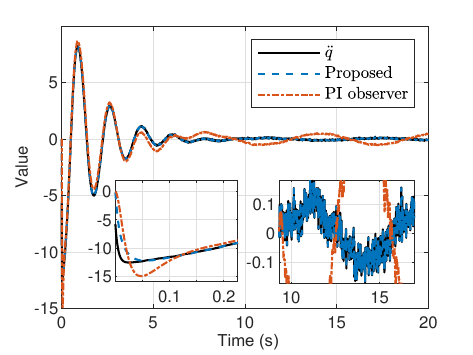}}
    \end{subfloat}
    \begin{subfloat}[]
        {\includegraphics[width=0.4\linewidth]{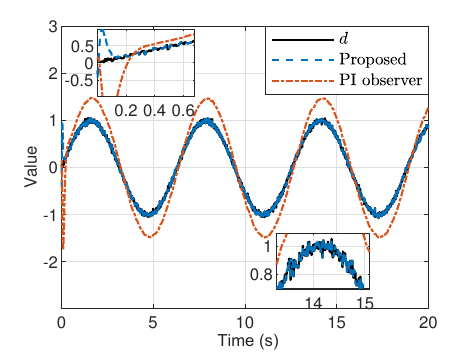}}
    \end{subfloat}
    \caption{Comparative state estimation results of the states (a) $q$, (b) $\dot{q}$, (c) $\ddot{q}$, and (d) the fault/disturbance $d$ for FAS model \eqref{eq:sim_FAS model} when the control input $u\equiv 0$. 
    The black solid line represents the actual signal, the blue dotted line indicates the estimation by the proposed observer \eqref{eq:oberver}, and the red dash-dot line corresponds to the estimation by the PI observer from \cite{JiangDuanHou2024TCSI}.}
    \label{fig:sim_estimation}
\end{figure}

\begin{figure}[htbp]
    \centering
    \begin{subfloat}[]
        {\includegraphics[width=0.4\textwidth]{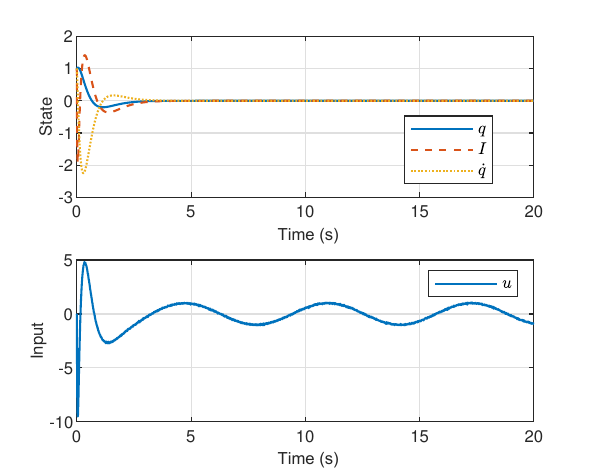}}
    \end{subfloat}~~
    \begin{subfloat}[]
        {\includegraphics[width=0.4\textwidth]{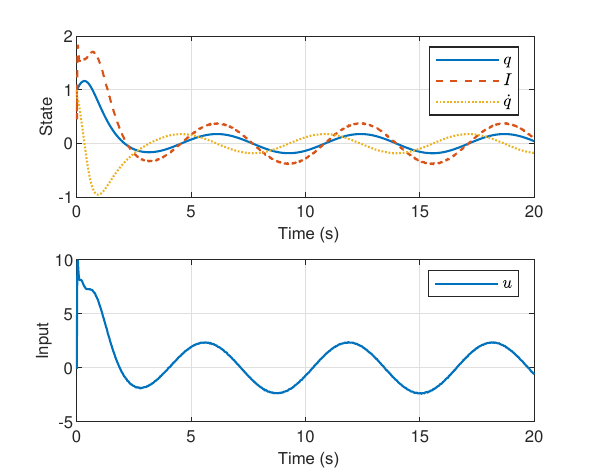}}
    \end{subfloat}
    \caption{The comparative control performance of the state response and control input of (a) the proposed method and (b) the method of \cite{JiangDuanHou2024TCSI}. 
    Note that the `state' presented in this figure is the original system \eqref{eq:sim_electromechanical system}'s state, i.e., $[q ~I ~\dot{q}]^{\mathrm{T}}$, instead of the FAS state $[q ~\dot{q} ~\ddot{q}]^{\mathrm{T}}$ that was presented in Figure \ref{fig:sim_estimation}.}
    \label{fig:sim_control}
\end{figure}

Then, Lemma \ref{lemma:parametric design approach} is used to design the controller gain. We specify the desired closed-loop system poles as $[-2 ~-3 ~-4]$ and select the free matrix $Z = [1 ~1 ~1]$. 
This produces the controller gain $K = [24 ~26 ~9]$.
Based on the proposed observer-based active fault/disturbance compensation control framework \eqref{eq:controller} with \eqref{eq:oberver}, we assess the control performance by comparing it with the method in \cite{JiangDuanHou2024TCSI}. The results are depicted in Figure \ref{fig:sim_control}. The comparison demonstrates that the active fault and disturbance compensation for the FAS model effectively stabilizes the system state. In contrast, the control performance of the method from \cite{JiangDuanHou2024TCSI} is less effective, exhibiting oscillatory behavior in the state.

A quantitative comparison of the proposed method and the existing method \cite{JiangDuanHou2024TCSI} is presented in Table \ref{tab:quantitative analysis of estimation and control performances} to highlight their performance in estimation and control.
For estimation performance, the proposed observer \eqref{eq:oberver} demonstrates superior accuracy. While the Root Mean Squared Error (RMSE) for position estimation error ($e_1$) is comparable, the proposed method achieves lower RMSE, Mean Absolute Error (MAE), and Maximum Absolute Error (Max AE) for the velocity ($e_2$), acceleration ($e_3$), and lumped disturbance estimation errors ($e_d$). This indicates a more precise and reliable estimation of the system dynamic states and unknown inputs.
The advantages in control performance are even more pronounced. The proposed controller reduces both the Integral of Absolute Error (IAE) and the Integral of Time-weighted Absolute Error (ITAE) for almost all state variables. In particular, the ITAE values are reduced by over $98\%$ for variables $q$ and $I$. The lower ITAE signifies that the proposed fault/disturbance compensation controller enhances tracking precision and ensures a much faster transient response, settling, and suppressed oscillations after a fault or disturbance occurs.

\begin{table}[htbp]
    \centering
    \caption{Quantitative analysis of estimation and control performances}
    \label{tab:quantitative analysis of estimation and control performances}
    \resizebox{\textwidth}{!}{%
        \begin{threeparttable}
            \begin{tabular}{@{}cccccccccccc@{}}
                \toprule
                \multicolumn{7}{c}{Estimation} & \multicolumn{5}{c}{Control} \\ 
                \cmidrule(lr){1-7} \cmidrule(lr){8-12}
                & \multicolumn{3}{c}{Proposed} & \multicolumn{3}{c}{Existing \cite{JiangDuanHou2024TCSI}} & & \multicolumn{2}{c}{Proposed} & \multicolumn{2}{c}{Existing \cite{JiangDuanHou2024TCSI}} \\ 
                \cmidrule(lr){2-4} \cmidrule(lr){5-7} \cmidrule(lr){9-10} \cmidrule(lr){11-12} 
                Variable & RMSE     & MAE     & Max AE  & RMSE     & MAE     & Max AE  & Variable & IAE           & ITAE         & IAE          & ITAE          \\
                $e_1$     & \textbf{0.0529}   & \textbf{0.0036}  & \textbf{1.0000}  & 0.0525   & 0.0058  & 1.0000  & $q$       & \textbf{0.6162}        & \textbf{0.4671}       & 3.3909       & 23.7725       \\
                $e_2$     & \textbf{0.0492}   & \textbf{0.0038}  & \textbf{1.0000}  & 0.1676   & 0.0770  & 2.7266  & $I$       & \textbf{1.4645}        & \textbf{0.9434}       & 3.3920       & 23.6701       \\
                $e_3$     & \textbf{0.2984}   & \textbf{0.0203}  & \textbf{5.6552}  & 0.5939   & 0.4017  & 8.6748  & $\dot{q}$ & \textbf{5.8474}        & \textbf{2.4436}       & 4.8802       & 25.4165       \\
                $e_d$     & \textbf{0.0493}   & \textbf{0.0050}  & \textbf{0.9786}  & 0.3560   & 0.3096  & 1.8890  &           &               &              &              &               \\ 
                \bottomrule
            \end{tabular}%
            \begin{tablenotes}
                \footnotesize
                \item Note 1: $e_1$, $e_2$, $e_3$, and $e_d$ are the corresponding estimation errors and are defined by $e_1 = q - \hat{q}$, $e_2 = \dot{q} - \dot{\hat{q}}$, $e_3 = \ddot{q} - \ddot{\hat{q}}$, and $e_d = d - \hat{d}$.
                
                \item Note 2: RMSE, MAE, and Max AE represent the Root Mean Squared Error, Mean Absolute Error, and Maximum Absolute Error, respectively. IAE and ITAE denote the Integral of Absolute Error and Integral of Time-weighted Absolute Error, respectively.
                
                \item Note 3: Since the variable-step solver is used in the simulation, we employ the high-precision numerical integration methods, such as the trapezoidal integration, to calculate the indices IAE and ITAE.
            \end{tablenotes}
        \end{threeparttable}
    }
\end{table}

\subsection{Experiment validation using Ball and Beam system}

The experiment is carried out using the Quanser Ball and Beam system \cite{Quanser_RotaryServo,Quanser_BallBeam}. The device is shown in Figure \ref{fig:exp_ball beam}. The overall experimental system consists of a Ball and Beam module, a Q8-USB data acquisition board, a VoltPAQ-X2 linear voltage-based power amplifier, and a lab computer.

\begin{figure}[htbp]
    \centering
    \includegraphics[scale=0.6]{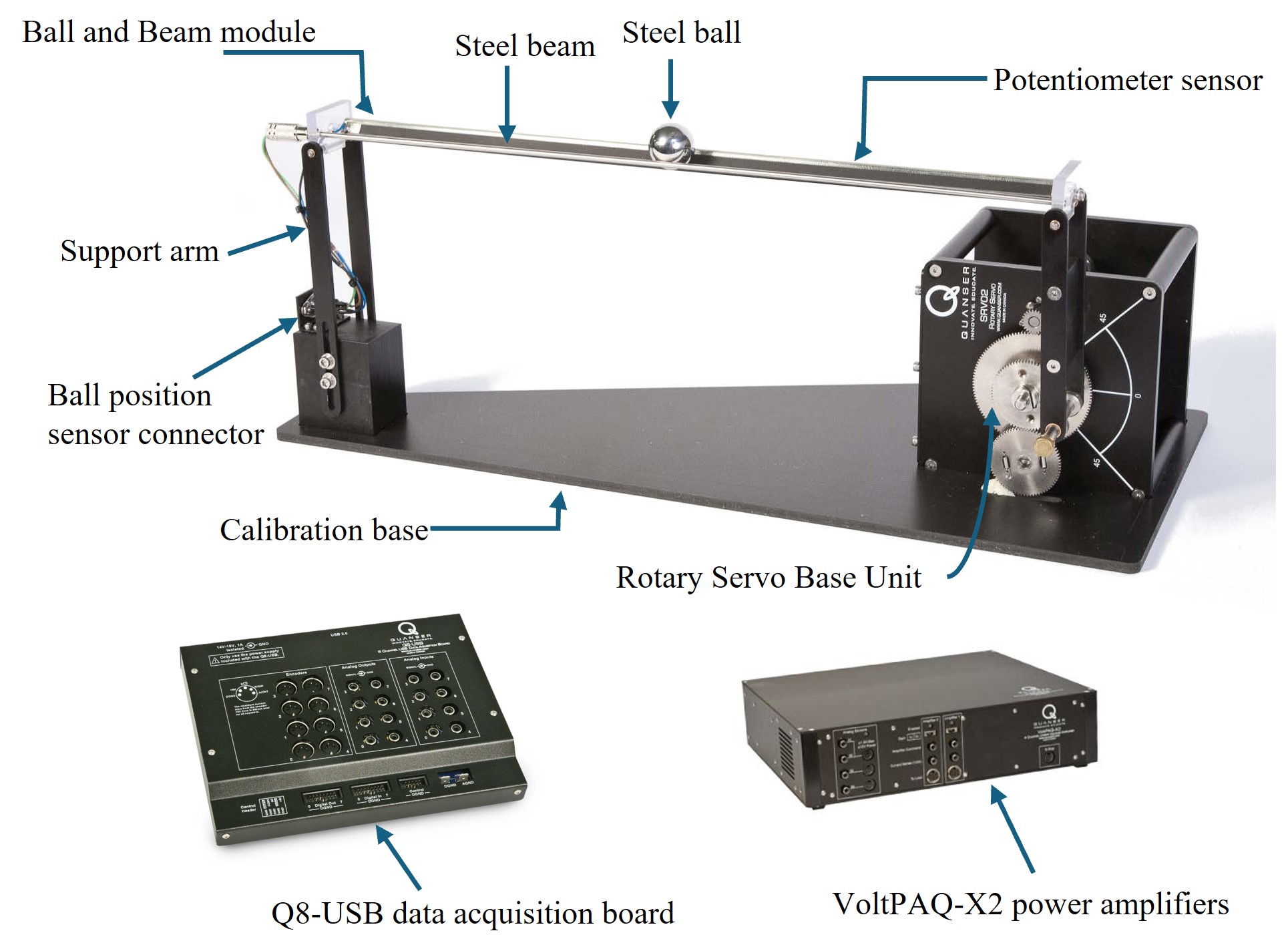}
    \caption{Quanser Ball and Beam system, motion data acquisition card, and power amplifier.}
    \label{fig:exp_ball beam}
\end{figure}

The second-order dynamic model of the Ball and Beam system with actuator fault/disturbance can be characterized by
\begin{align}	\label{eq:exp_ball beam system}
    \left\{ \begin{array}{l}
        \ddot{z}=K_{bb}\sin \theta _l\\
        \ddot{\theta}_l=-\frac{1}{\tau_s}\dot{\theta}_l+\frac{K_s}{\tau_s}(v_m+d),
    \end{array} \right.
\end{align}
which falls under the framework of second-order Type III underactuated systems, as generalized in the work of \cite{Duan_IJSS_FourTypes}. The variables $z$ and $\theta_l$ represent the ball displacement and servo angle, respectively, $v_m$ is the motor input voltage. The model gain is given by $K_{bb}=\frac{m_bgr_br_{\mathrm{arm}}}{m_br_{b}^{2}L_b+J_bL_b}$, where the moment of inertia of a solid sphere is $J_b = \frac{2}{5}m_br_b^2$. Other parameters with corresponding descriptions and values are detailed in Table \ref{table:Ball and Beam parameters}. In this system, two outputs, $z$ and $\theta_l$, are measured using an analog potentiometer sensor and an optical encoder, respectively. To simulate challenging external conditions, a fault signal is deliberately injected into the analog potentiometer sensor, indicating that the collected ball position data is inaccurate or even messed up. The measurement equations are given by $y_1 = z + 0.1d$ and $y_2 = \theta_l$.

\begin{table}[t]
    \centering
    \caption{Parameter descriptions of Ball and Beam module}
    \label{table:Ball and Beam parameters}
    \begin{tabular}{ccc}
        \toprule
        Parameter & Description & Value \\
        \midrule
        $m_b$ & Mass of ball & $0.064 \mathrm{kg}$ \\
        $r_b$ & Radius of ball & $0.0127 \mathrm{m}$ \ \\
        $r_\mathrm{arm}$ & \makecell[c]{Distance between servo output\\ gear shaft and coupled joint} & $0.0254 \mathrm{m}$ \\
        $L_b$ & Beam length & $0.4255 \mathrm{m}$ \\
        $K_s$ & Steady-state gain of rotary servo base unit & $1.5 \mathrm{rad/s/V}$ \\
        $\tau$ & Time constant of rotary servo base unit & $0.025 \mathrm{s}$ \\
        $g$ & Gravitational constant & $9.8 \mathrm{m/s^2}$ \\
        \bottomrule
    \end{tabular}
\end{table}

In the literature of the modern control field, the dynamics model \eqref{eq:exp_ball beam system} is frequently transformed into a state-space representation to facilitate the application of control or estimation techniques. However, rather than using the state-space framework, we address the observation and control objectives through the FAS approach. The following procedures provide details about converting \eqref{eq:exp_ball beam system} into a FAS model. For the sake of simplicity in the subsequent derivations, we denote $\epsilon _1=K_{bb}$, $\epsilon _2=-\frac{1}{\tau_s}$, $\epsilon _3=\frac{K_s}{\tau_s}$, and $u = v_m$. According to the standard framework of Type III underactuated systems defined in \cite{Duan_IJSS_FourTypes}, \eqref{eq:exp_ball beam system} can be written into
\[
\left\{ \begin{array}{l}
    \ddot{z}=g\left( \theta _l,t \right)\\
    \ddot{\theta}_l=\tilde{u},
\end{array} \right.
\]
where $g\left( \theta _l,t \right) = \epsilon _1\sin \theta _l$, $\tilde{u} = \epsilon _2\dot{\theta}_l+\epsilon _3(u+d)$, and a standard transformation is provided as follows:

\begin{lemma}   \label{lemma:Duan IJSS FourTypes_Type III UASs}
    \cite{Duan_IJSS_FourTypes}
    Consider the second-order Type III underactuated system (UAS)
    \begin{align}   \label{eq:lemma second-order Type III UASs}
        \begin{cases}
            \ddot{q}_1=u\\
            \ddot{q}_2=g\left( q_1,q_{2}^{\left( 0\sim 1 \right)},t \right).
        \end{cases}
    \end{align}
    Under the following mapping
    \[
    \begin{cases}
        q_1=h_1\left( q^{\left( 0\sim 3 \right)},t \right)\\
        \dot{q}_1=h_2\left( q^{\left( 0\sim 3 \right)},t \right)\\
        q_2=q,
    \end{cases}
    \]
    the UAS \eqref{eq:lemma second-order Type III UASs} can be transformed into a FAS as
    \[
    q^{(4)}=f\left( q^{(0\sim 3)},t \right) +B\left( q^{(0\sim 3)},t \right) u,
    \]
    where
    \begin{align*}
        &f\left( q^{(0\sim 3)},t \right) =f_0\left( h_1\left( \cdot \right) ,h_2\left( \cdot \right) ,q^{\left( 0\sim 3 \right)},t \right),\\
        &B\left( q^{(0\sim 3)},t \right) =B_0\left( h_1\left( \cdot \right) ,h_2\left( \cdot \right) ,q^{\left( 0\sim 3 \right)},t \right),
    \end{align*}
    
    \begin{align*}
        &f_0\left( h_1\left( \cdot \right) ,h_2\left( \cdot \right) ,q^{\left( 0\sim 3 \right)},t \right)\\
        &=\left\{ \frac{\partial ^2g}{\partial q_{1}^{2}}\dot{q}_1+\frac{\partial ^2g}{\partial q_1\partial q_{2}^{(0\sim 1)}}q_{2}^{(1\sim 2)}+\frac{\mathrm{d}}{\mathrm{d}t}\frac{\partial g}{\partial q_1} \right\} \dot{q}_1
        +\left\{ \frac{\partial ^2g}{\partial q_2\partial q_1}\dot{q}_1+\frac{\partial ^2g}{\partial q_2\partial q_{2}^{(0\sim 1)}}q_{2}^{(1\sim 2)}+\frac{\mathrm{d}}{\mathrm{d}t}\frac{\partial g}{\partial q_2} \right\} \dot{q}_2\\
        &+\left\{ \frac{\partial ^2g}{\partial \dot{q}_2\partial q_1}\dot{q}_1+\frac{\partial ^2g}{\partial \dot{q}_2\partial q_{2}^{(0\sim 1)}}q_{2}^{(1\sim 2)}+\frac{\mathrm{d}}{\mathrm{d}t}\frac{\partial g}{\partial \dot{q}_2} \right\} \ddot{q}_2
        +\frac{\partial g}{\partial q_2}\ddot{q}_2+\frac{\partial g}{\partial \dot{q}_2}\dddot{q}_2+\frac{\partial ^2g}{\partial t^2}\\
        &=f_0\left( q_{1}^{\left( 0\sim 1 \right)},q_{2}^{\left( 0\sim 3 \right)},t \right),
    \end{align*}
    and
    \begin{align*}
        B_0\left( h_1\left( \cdot \right) ,h_2\left( \cdot \right) ,q^{\left( 0\sim 3 \right)},t \right) 
        =\frac{\partial g}{\partial q_1}u=B_0\left( q_{1}^{\left( 0\sim 1 \right)},q_{2}^{\left( 0\sim 3 \right)},t \right).
    \end{align*}
\end{lemma}

Therefore, following the standard framework for Type III UASs in Lemma \ref{lemma:Duan IJSS FourTypes_Type III UASs}, we finally achieve the FAS model of the Ball and Beam system \eqref{eq:exp_ball beam system} below, as outlined earlier in \eqref{eq:big FAS model} or \eqref{eq:compact FAS model}:
\begin{align}   \label{eq:FAS model of Ball and Beam system}
    \left\{ \begin{array}{l}
        x^{\left( 4 \right)}=f\left( x^{\left( 0\sim 3 \right)} \right) +B\left( y \right) u+D_1(y)d\\
        y=Cx^{\left( 0\sim 3 \right)}+D_2d,
    \end{array} \right.
\end{align}
where
\[
f\left(  x^{\left(  0\sim3\right)  }\right)  =\epsilon _2\dddot{x}-\frac{\dddot{x}^2\ddot{x}}{1-\ddot{x}^2},
~B(y) = D_1(y)  =\epsilon _3\sqrt{1-\ddot{x}^2},
~C = \left[ \begin{matrix}
    1&		0&		0&		0\\
    0&		0&		1&		0\\
\end{matrix} \right],
~D_2=\left[ \begin{array}{c}
    0.1\\
    0\\
\end{array} \right],
\]
under the introduced diffeomorphism transformation:
\begin{equation}
    \begin{cases}
        z=\epsilon _1x\\
        \dot{z}=\epsilon _1\dot{x}\\
        \theta _l=\mathrm{arcsin} \left( \ddot{x} \right)\\
        \dot{\theta}_l=( \sqrt{1-\ddot{x}^2} ) ^{-1}\dddot{x}.
    \end{cases}
    \label{eq:exp_inverse diffeomorphism transformation}
\end{equation}
The sensors provide measurement information of the ball position $z$ and the servo angle $\theta_l$. Since we can obtain $x$ and $\ddot{x}$ ($x = z$ and $\ddot{x} = \epsilon _1\sin \theta _l$) from sensors, as given in \eqref{eq:exp_inverse diffeomorphism transformation}, the measurement equation for the FAS model can be formulated accordingly in \eqref{eq:FAS model of Ball and Beam system}.

The controller for the Ball and Beam system \eqref{eq:exp_ball beam system} can be easily designed based on the obtained FAS model \eqref{eq:FAS model of Ball and Beam system} as designed in \eqref{eq:controller}, as long as 
$
\det B(y) \ne 0 \text{ or } \infty,
$
i.e., 
servo angle $\theta_l$ must operate within the specified constrained region
$
\theta _l\in [ -\frac{\pi}{2},\frac{\pi}{2} ] (\mathrm{rad}),
$
to ensure system stability, prevent mechanical failure, and maintain optimal performance. Exceeding the bound could lead to undesirable behavior, such as changes in the model and deviations from the intended control objectives. Therefore, as stated in \cite{Duan_IJSS_FourTypes}, sometimes a pre-controller is needed to drive the system trajectory back to the safety region or region of exponential attraction.

According to Theorem \ref{theorem:observer}, by letting $\mu_e = 8$, $\gamma_f = 5$ and solving LMI \eqref{eq:observer LMI}, we obtain $\eta = 0.5305$ and the following observer parameters
\[
T = \begin{bmatrix}
    1 & 0 & 14.1890 & 0 & 0 \\
    0 & 1 & 116.3720 & 0 & 0 \\
    0 & 0 & -1.6509 & 0 & 0 \\
    0 & 0 & -91.8945 & 1 & 0 \\
    -10 & 0 & -124.5768 & 0 & 0
\end{bmatrix},~
N = \begin{bmatrix}
    0 & -14.1890 \\
    0 & -116.3720 \\
    0 & 2.6509 \\
    0 & 91.8945 \\
    10 & 124.5768
\end{bmatrix},~
L = \begin{bmatrix}
    9.4393 & 1.4323 \\
    -9.2283 & 16.2255 \\
    0.2110 & 10.5871 \\
    13.4029 & -14.0984 \\
    11.2714 & -16.4543
\end{bmatrix}.
\]

Then, the FAS controller parameter can be calculated by utilizing Lemma \ref{lemma:parametric design approach}. With the prescribed poles $[-3 ~-4 ~-0.8 ~-0.9]$ and the free matrix $Z = [1 ~1 ~1 ~1]$, we obtain the control gain
\[
K = \begin{bmatrix}
    8.64 & 25.44 & 24.62 & 8.7
\end{bmatrix}.
\]

The initial values for the observer \eqref{eq:oberver} are set to $0.6667\times [1 ~1 ~1 ~1 ~1]^{\mathrm{T}}$. The experiment considers two scenarios, each involving a different type of signal $d$. In scenario 1, a sinusoidal fault/disturbance signal $d = 6\sin(t)$ is injected into the servo motor input $V_m$ and the analog potentiometer sensor (ball position sensor) with corresponding coefficients. In scenario 2, a rapidly varying $d$, composed of a sine wave with magnitude 6 and fast time-varying frequency, combined with a square wave with magnitude 2, is applied. For each scenario, we observe the state responses of ball position $z$ and servo angle $\theta_l$ to determine if the states stabilize to zero, and simultaneously assess the practicability of fault/disturbance estimation using the proposed observer \eqref{eq:oberver}. To be more convincing, results without active compensation, where the term $\hat{d}$ is removed from FAS controller \eqref{eq:controller}, are also presented.

The overall experiment results are shown in Figure \ref{fig:exp_control}. The experimental data clearly demonstrate that the proposed active compensation control for the FAS model is practically feasible and maintains good performance, even under severe faults or disturbances in actuators and core sensors (i.e., ball position sensor in an underactuated mechanical structure). In contrast, the control without compensation fails completely, leading to divergent system states, let alone achieving good control performance. The estimation results validate the ability of the proposed observer to reconstruct signals rapidly and accurately, even when the fault or disturbance signal varies quickly. 

\begin{figure}[htbp]
    \centering
    \subfloat[]
        {\includegraphics[width=0.5\textwidth]{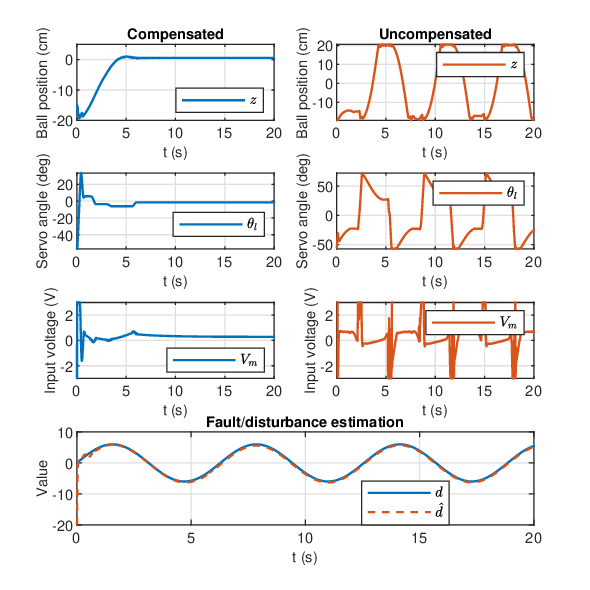}}
    \subfloat[]
        {\includegraphics[width=0.5\textwidth]{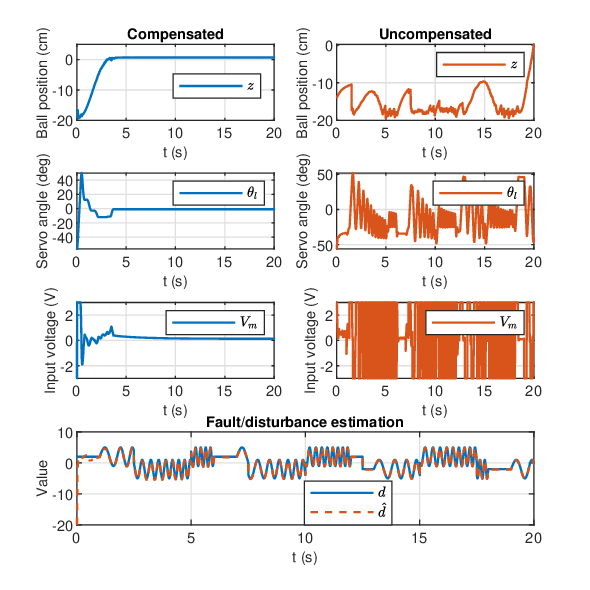}}
    \caption{Experimental results. (a) Scenario 1; (b) Scenario 2.
    Each subfigure contains: state responses for ball position $z$, servo angle $\theta_l$, and control input $V_m$ using the proposed method with active compensation (top left); corresponding results without active compensation (top right); and the estimation of the fault/disturbance signal $d$ (bottom).}
    \label{fig:exp_control}
\end{figure}

\section{Conclusion}	\label{sec:Conclusion}
In this paper, we developed the observer-based active fault/disturbance compensation control for continuous-time nonlinear FASs. By introducing an observer with enhanced design flexibility, we achieved a precise estimation of system states and fault/disturbance signals. 
The comparative experimental results confirmed the feasibility and superiority of the method, demonstrating its potential for real-world applications in nonlinear system control.
In future work, we intend to extend the proposed method to the unidirectionally connected FASs \cite{Duan2025IJSS_UC-FAS:I}. Another future research direction involves exploring set-membership estimation techniques \cite{RenDuanFuKong2026TAES} and interval observers \cite{Ren2024FASTA_IntervalObserver,RenDuanLiKong2025TMech} within the FAS framework for advanced fault diagnosis purposes.

\Acknowledgements{
    The authors express gratitude to Dr. Ruirui Yang for useful comments and discussions, and also thank Mr. Yifan Wang and Mr. Jianpeng Zou for their proofreading assistance.
    This work was supported in part by the National Natural Science Foundation of China (NSFC) under Grant 623B2045 and the Science Center Program of NSFC under Grant 62188101, in part by the Guangdong Science and Technology Program under Grant 2024B1212010002, in part by the Shenzhen Science and Technology Program under Grant KQTD20221101093557010.
    This work was also supported in part by the NSFC under Grants 62203206, 62573221, and 62550003, in part by the Guangdong Provincial Natural Science Foundation under Grant 2024A1515011648, in part by the Shenzhen Science and Technology Program under Grants JCYJ20240813094403005 and JCYJ20240813094212017, and in part by the Key Program of Special Funds for the Cultivation of Guangdong College Students Scientific and Technological Innovation (``Climbing Program'' Special Funds) under funding pdjh2025ak186.
    
}

\end{document}